\documentclass[a4paper,11pt,
reqno]{amsart}

\usepackage[utf8]{inputenc}
\usepackage[T1]{fontenc}
\usepackage[english]{babel}
\usepackage{amsmath,amssymb,amsthm,mathrsfs,dsfont}
\usepackage{mathtools}
\usepackage{todonotes}
\usepackage[shortlabels]{enumitem}
\usepackage{csquotes}
\MakeOuterQuote"
\usepackage[colorlinks]{hyperref}
\usepackage[capitalize]{cleveref}

\newcommand\includechanges{ \input{changes.tex} }

\renewcommand\includechanges{}

\newtheorem{thm}{Theorem}
\newtheorem*{con}{Conjecture}
\newtheorem{lem}[thm]{Lemma}
\newtheorem{cor}[thm]{Corollary}
\newtheorem{prop}[thm]{Proposition}
\theoremstyle{definition}

\newtheorem{rem}[thm]{Remark}
\newtheorem{ex}{Example}\crefname{ex}{Example}{Examples}

\DeclareMathOperator\lh{lh}

\newcommand\1{\mathds{1}}
\newcommand\abs[1]{\lvert #1 \rvert}
\newcommand\BO[1]{\mathcal{B}(#1)}
\renewcommand{\bar}[1]{\overline{#1}}

\newcommand\CC{\mathbb C}
\renewcommand\d{\,\textup{d}}

\newcommand\ip[2]{\langle#1,#2\rangle}
\renewcommand\Im{\operatorname{Im}}
\renewcommand\H{\mathcal H}
\newcommand\K{\mathcal K}
\newcommand\norm[1]{\lVert#1\rVert}

\newcommand\NN{\mathbb N}

\newcommand\placeholder{\,\cdot\,}
\newcommand\RR{\mathbb R}
\renewcommand\Re{\operatorname{Re}}

\newcommand\Sch{\mathscr S}
\newcommand\SE{\mathcal S}
\newcommand\operatorvalued{operator-valued}
\newcommand\vectorvalued{vector-valued} 
\newcommand\VectorValued{Vector-Valued}
\newcommand\HK{\H(\K)}

\newcommand\SEK{\SE(\K)}

\title[Self-Adjointness of Toeplitz Operators on the Segal-Bargmann Space]{Self-Adjointness of Toeplitz Operators \\ on the Segal-Bargmann Space}

\author[W. Bauer]{Wolfram Bauer}
\address{Wolfram Bauer, Leibniz Universität Hannover, Institut für Analysis, Welfengarten 1, 30167 Hannover, Germany }
\email{bauer@math.uni-hannover.de}
\author[L. van Luijk]{Lauritz van Luijk}
\address{Lauritz van Luijk, Institut für Theoretische Physik, Leibniz Universität Hannover, Appelstraße 2, 30167 Hannover, Germany}
\email{lauritz.vanluijk@itp.uni-hannover.de}
\author[A. Stottmeister]{Alexander Stottmeister}
\address{Alexander Stottmeister, Institut für Theoretische Physik, Leibniz Universität Hannover, Appelstraße 2, 30167 Hannover, Germany}
\email{alexander.stottmeister@itp.uni-hannover.de}
\author[R. F. Werner]{Reinhard F. Werner}
\address{Reinhard F.~Werner, Institut für Theoretische Physik, Leibniz Universität Hannover, Appelstraße 2, 30167 Hannover, Germany}
\email{reinhard.werner@itp.uni-hannover.de}

\date{\today}

\begin{document}

\includechanges

\begin{abstract}
    We prove a new criterion that guarantees self-adjointness of Toeplitz operators with unbounded \operatorvalued\ symbols.
    Our criterion applies, in particular, to symbols with Lipschitz continuous derivatives, which is the natural class of Hamiltonian functions for classical mechanics.
    For this we extend the Berger-Coburn estimate to the case of \vectorvalued\ Segal-Bargmann spaces.
    Finally, we apply our result to prove self-adjointness for a class of (\operatorvalued) quadratic forms on the space of Schwartz functions in the Schr\"odinger representation. 
\end{abstract}
\maketitle


\section{Introduction}\label{sec:intro}

Proving the self-adjointness of unbounded Hilbert space operators is a key problem for many applications of functional analysis in quantum theory.
For von Neumann, who founded the theory \cite{vN}, the principal distinction of self-adjoint operators over merely hermitian operators was the existence of a spectral resolution. 
This property was significant to him as the creator of the mathematical formalism of quantum theory, not least because a partition of unity into projections was needed for a quantum-probabilistic interpretation as a quantum observable. 
Equally important in modern applications is the role of self-adjoint operators as generators of dynamical evolution. According to the rules of quantum mechanics, this must be given by a strongly continuous unitary one-parameter group and, hence, be generated by a self-adjoint operator. 
A hermitian operator built as a sum of terms motivated by a physical model may very well fail this crucial property. 
This is then a sign that the problem has not been sufficiently well defined, and a self-adjoint extension must be singled out, typically by specifying suitable boundary conditions.
Once self-adjointness is achieved, and so the generator can be exponentiated, the group property extends this to all times.
The counterpart of the self-adjointness problem for classical dynamical systems is, therefore, the \emph{global} existence of the flow (possibly only for almost all initial values, as the Coulomb potential $V(x) = \abs x^{-2}$ shows). This connection is not just superficial and can be studied case by case. 
Indeed, a folk ``theorem'' states that a Schr\"odinger operator with given potential is self-adjoint on its natural domain precisely if the Hamiltonian flow for the corresponding problem of classical mechanics exists for all times. 
There are well-known counterexamples to this statement in either direction \cite[Ch.\ X.1]{reed1975methods}. Nevertheless, there is some truth in it, so the question may be which additional regularity assumptions make it true. This paper establishes a criterion for self-adjointness directly motivated by a classical counterpart. 

The basic connection between the classical and the quantum world is provided by a quantization $f\mapsto T_f$ relating classical phase space functions $f$ to quantum operators $T_{f}$, e.g., the classical Hamilton function to the quantum Hamiltonian.
Even though the natural stage for quantum mechanical operators $T_{f}$ of systems with $n$ canonical degrees of freedom is without a doubt the Schr\"odinger representation, i.e., the Hilbert space $L^2(\RR^n)$ together with the standard representation of the canonical commutation relations in terms of position and momentum operators. We will, however, work in the unitarily equivalent Segal-Bargmann representation.
The main reason for this is that we will use the coherent-state quantization which takes on a more natural form in this representation, known as the \emph{Berezin-Toeplitz quantization}.
It depends on a quantization parameter $t>0$, which plays the role of Planck's constant. 
Smaller values of $t$ correspond to larger scales of \emph{action} and the limit $t\to 0$ (formally) is the classical limit.
A natural criterion to ensure global existence of the solutions to Hamilton's equations is provided by the Picard-Lindel\"of theorem, which requires the derivatives of the classical Hamiltonian function $f$ to be globally Lipschitz continuous.
We will show that precisely the same assumption guarantees self-adjointness of the Berezin-Toeplitz operator $T_f$.
As it turns out, we can even weaken the assumption of Lipschitz continuity in that we only need the derivatives of a heat transformed symbol to have bounded oscillation. 
This comes in handy, as the heat transformed symbol is generally better behaved. 

In the following we denote by $T_f$ the Toeplitz operator with symbol $f$ acting on the Segal-Bargmann space $\mathcal{H}$ of Gaussian square integrable analytic functions (s. \cref{sec:notation}). 
Applying the well-known reproducing kernel function of $\mathcal{H}$ we express $T_f$ in an integral form and consider it on its natural domain of functions for which this integral again defines an element in the Segal-Bargmann space (see \eqref{eq:toeplitz_domain} for details).

\begin{thm}\label{thm:mainscalar}
    Fix a quantization parameter $t>0$.
    Let $f :\RR^{2n}\!\cong\CC^n \to \RR$ be a measurable polynomially bounded symbol such that there is a time $s\in [0,\frac t2)$ for which the first-order derivatives of the heat-transformed symbol $\widetilde{f}^{(s)}$  at time $s$ have bounded oscillation, i.e., satisfy
    \begin{equation*}
        \abs{\partial_j \widetilde f^{(s)}(z) -\partial_j \widetilde f^{(s)}(w) } \lesssim 1+\abs{z-w}, \quad j=1,\ldots,2n,
    \end{equation*} 
    for all $z,w \in\CC^n$. 
    Then the Toeplitz operator $T_f$ is self-adjoint on its natural domain.
\end{thm}

The symbol ``$\lesssim$'' means that the inequality ``$\leq$'' holds provided that the right-hand side is multliplied by some constant $c>0$, and $\widetilde f^{(s)}$ is the heat transform of $f$ after time $s$.
We will extend this theorem to allow for \operatorvalued\ symbols and use it to prove a self-adjointness criterion for (\operatorvalued) continuous quadratic forms on the space of Schwartz functions in the Schr\"odinger representation (see Theorems \ref{thm:main} and \ref{thm:forms_schrodinger}). 
One ingredient of the proof are the Berger-Coburn estimates \cite{bauer2020berger,berger1994heat} which we derive in the case of \operatorvalued\ symbols in \cref{sec:bce}.

Previous work on the self-adjointness of unbounded Toeplitz operators was done by J.\ Janas in \cite{janas1,janas3,janas2}.
Apart from considering the scalar-valued case only (i.e., no \operatorvalued\ symbols), his work differs from ours in several aspects.
In particular, Janas investigates many aspects of unbounded operators, e.g., closedness or computation of the adjoint, not just self-adjointness.
While Janas' work often considers families of symbols with specific properties such as (anti-)analyticity, we are interested in classes of symbols that also serve as natural Hamiltonian functions for a classical mechanical system.
In \cite[Thm.\ 3.3]{janas3} it is proved (for $t=\frac12$) that the Toeplitz operators with measurable real-valued symbols $f$ are self-adjoint if 
\[
    \abs{f(z)-f(w)} \lesssim e^{\frac 1{8t}\abs{z-w}^2}
\] 
for all $z,w \in \CC^n$. 
By noting that this inequality still holds if the right-hand side is replaced by $1+\abs{z-w}$, we will show that \cref{thm:mainscalar} is a generalization of this result (see \cref{rem:janas}). 

Finally, we collect some examples and applications in \cref{sec:examples}. 
Among others, these include results on the self-adjointness of relativistic and non-relativistic Schr\"odinger operators with operator-valued potentials, which are new to the authors best knowledge.

\section{Notation}\label{sec:notation}

We denote by $\H$ the Segal-Bargmann space with quantization parameter $t>0$,
which is defined as the closed subspace consisting of entire functions in $L^2(\CC^n,\mu)$, where 
\[
    \d\mu(z) = (2\pi t)^{-n} e^{- \frac{\abs z^2}{2t}}\d z.
\] 
The norm and inner product on $\H$ are induced by $L^2(\CC^n,\mu)$ and will be denoted by $\norm\placeholder$ and $\ip\placeholder\placeholder$, respectively.
We will use these symbols for the scalar products and norms of all Hilbert spaces, since it will be clear from the context to which space the vectors belong. Nevertheless, we will sometimes append the space as an index for emphasis when that seems helpful.
$\H$ is a reproducing kernel Hilbert space with kernel $K(z,w) = e^{\frac{\bar w \cdot z}{2t}}$ where $z,w \in\CC^n$, i.e., $f(z) = \ip f{K(\placeholder,z)}$ for all $f\in\H$.
The normalized kernels are denoted by $k_w = \norm{K(\placeholder,w)}^{-1} K(\placeholder,w)$.
The standard orthonormal basis of $\H$ is given by the monomials
\[
    e_\nu(z) = \frac{z^\nu}{\sqrt{\nu! (2t)^{\abs\nu}}}, 
    \quad \nu\in \NN_0^n,
\] 
where $z^\nu = (z_1)^{\nu_1}\cdots (z_n)^{\nu_n}$, $\nu! = \nu_1!\cdots \nu_n!$ and $\abs\nu = \nu_1+\dots+\nu_n$.
Denote by $P : L^2(\CC^n,\mu) \to \H$ the orthogonal projection onto $\H$.
The projection $P$ can be expressed in terms of the reproducing kernel by
\[
    Pf(z) = \int K(z,w) f(w)  \d\mu(w).
\]
For a measurable function $f:\CC^n \to \CC$, the \emph{Toeplitz operator} acts on a suitable domain in $\H$ by
\begin{equation}\label{eq:toeplitz_pmf}
    T_fg = P(f\cdot g).
\end{equation} 
Using the reproducing kernel this leads to 
\begin{equation}\label{eq:toeplitz_kernel}
    T_fg(z) = \int K(z,w) f(w) g(w) \d\mu(w).
\end{equation} 
When denoting Toeplitz operators we will use ``$z$'' as a dummy variable, i.e., $T_{f(z)}=T_f$.

The unitary \emph{Weyl-operators} act by 
\begin{equation*}\label{eq:weyl_def}
    W_w f (z) = k_w(z) \, f(z-w),
\end{equation*} 
where $w,z\in\CC^n$.
The reproducing kernel is connected to the Weyl-operators through $k_z = W_z1$.
The Weyl-operators are generated by unbounded Toeplitz operators with linear symbols
\begin{equation}\label{eq:weyl_exponential}
    W_w = e^{\frac1{2t}( T_{w \cdot \bar z} - T_{\bar w\cdot z})} = e^{\frac it T_{\omega(w,z)}} 
\end{equation}
where $\omega(z,w) =  \Im(\bar z\cdot w)$.
The mapping $z \mapsto W_z$ is a strongly continuous, projective unitary representation of the group $(\CC^n,+)$ of phase space translations.
They satisfy the so-called \emph{Weyl-relations}, 
$ W_w W_z =  e^{\frac i{2t} \omega(z,w)} W_{z+w} $ 
and, hence, define a proper unitary representation of the Heisenberg group $\mathbb H_n$. 
The Weyl-operators implement phase space translations on $\H$ \cite{coburn1999measure}. 
To see this, we define $\alpha_z (A)= W_z A W_{-z}$ and $\alpha_z f(w) = f(w+z)$ for operators $A$ on $\H$ and functions $f$ on $\CC^n$. 
Then it follows that
\[
    \alpha_z T_f = T_{\alpha_z f}.
\] 

The \emph{Berezin transform} of an operator $A:D(A)\to \H$ is $\widetilde A(z) = \ip{A k_z}{k_z}$, provided $k_z \in D(A)$ for all $z\in\CC^n$.
While the Toeplitz operator is a form of covariant quantization, the Berezin transform provides a covariant dequantization, i.e., $\alpha_z \widetilde A = (\alpha_z A)^\sim$.
The Berezin transform of a Toeplitz operator $T_f$ is the heat transform of the symbol $f$ after time $t$ \cite{berger1994heat}.

By the \emph{Schr\"odinger representation} we mean the Hilbert space $L^2(\RR^d)$, which is isometrically isomorphic to $\H$ via the  \emph{Bargmann isometry} (cf.\ \cite[Chap.\ 1.6]{folland2016harmonic})
\[
    B : \H \to L^2(\RR^d).
\] 
The Bargmann isometry maps the standard basis $\{e_\nu\}$ to the basis of (scaled) Hermite functions:
$ B e_\nu (x) =  t^{-n /4} \psi_\nu(x /\sqrt t) $,
where $\psi_\nu$ is the $\nu$-th Hermite function.

Let $\K$ be a separable complex Hilbert space.
The \emph{\vectorvalued\ Segal-Bargmann space} $\HK$ is the closed subspace of analytic functions in $L^2(\CC^n,\mu;\K)$, the space of square integrable $\K$-valued functions with inner product
\[
    \ip fg_{L^2(\CC^n,\mu;\K)} = \int \ip{f(z)}{g(z)} \d \mu(z)
\]
for $f,g\in L^2(\CC^n,\mu;\K)$.
The norm and inner product on $\HK$ are, as in the scalar-valued case, induced by $L^2(\CC^n,\mu;\K)$ and are again denoted $\norm\placeholder$ and $\ip\placeholder\placeholder$, respectively.
The space $\HK$ can also be viewed as the Hilbert space tensor product $\H \otimes\K$ by setting $(f\otimes x) (z) \coloneqq  f(z)x$.
We denote by $\BO{\mathcal{V}}$ the algebra of bounded operators on a Hilbert space $\mathcal V$.
For a weakly measurable (see \cref{rem:measurability}) symbol $f : \CC^n \to \BO\K$, the Toeplitz operator $T_f$ is defined on a suitable domain as in the scalar-valued case:
$T_fg \coloneqq  P(f\cdot g)$, $g\in \HK$, where $(f\cdot g)(z) = f(z)g(z)$ and $P:L^2(\CC^n,\mu;\K) \to \HK$ is the orthogonal projection onto $\HK$.
The integral in \eqref{eq:toeplitz_kernel} exists as a Bochner-integral in $\K$.
In the case that $f = \lambda \otimes A$, i.e., $f(z)x = \lambda(z) Ax$, $x\in\K$, with $\lambda \in L^\infty(\CC^n)$ and $A \in \BO\K$, one has $T_f = T_\lambda \otimes A$.


\section{Schwartz Elements}\label{sec:schw}

In \cite{bargmann1967hilbert} the set of analytic functions $f\in\H$ that are mapped to Schwartz-functions by the Bargmann isometry, i.e., $Bf \in \Sch(\RR^n)$, was characterized by a growth condition (cf.\ \cref{thm:schwartz_elts}.\ref{it:bargmann}).
We refer to such functions as \emph{Schwartz elements} of the Segal-Bargmann space.

The space of Schwartz elements in $\H$, which we denote by $\SE$, will serve us as an invariant domain for a class of unbounded Toeplitz operators and as a form domain for quadratic forms.
Especially for the latter it is convenient (for our purposes) that $\SE$ is a Fr\'echet space with respect to a natural topology, and that $\SE$ is directly connected to the action of the Weyl-operators (cf.\ \cref{thm:schwartz_elts}.\ref{it:weyl}).

The following proposition is a collection of several equivalent definitions of the space $\SE$.

\begin{prop}\label{thm:schwartz_elts}
    Let $f \in \H$. Then the following are equivalent
    \begin{enumerate}[(1)]
        \item\label{it:schrodinger}
            $f \in \SE$, i.e., $B f \in \Sch(\RR^d)$,
        \item\label{it:seq}
            $(\ip f{e_\nu})_{\nu \in \NN_0^n}$ is rapidly decreasing%
            ,
        \item\label{it:bargmann}
            $\abs{f(z)} \abs z^N \lesssim e^{\frac1{4t}\abs z^2}$ for all $N\in \NN_0$,
        \item\label{it:FT}
            $f(z)\,e^{-\frac 1{4t}\abs z^2} \in \Sch(\CC^n)$,
        \item\label{it:pol}
            $z^\alpha \partial^\beta f(z) \in \H$ for all $\alpha,\beta \in \NN_0^n$,
        \item\label{it:weyl}
            $\CC^n \ni  z \mapsto W_z f \in \H$ is smooth.
    \end{enumerate}
\end{prop}

For the proof, we need the following lemma which follows readily from the fact that $(W_{sw})_{s\in\RR}$ is a strongly continuous unitary one-parameter group.
We denote by $\nabla f (z)$ and $\bar\nabla f(z)$ the vectors $(\partial_1 f(z), \dots, \partial_n f(z))$ and $(\bar\partial_1 f(z), \dots, \bar\partial_n f(z))$, respectively. Furthermore, we use the notation $\bar z= (\bar z_1,\dots,\bar z_n)$.

\begin{lem}\label{thm:domain_toep}
    Let $w\in \CC^n$. Then $s\mapsto W_{sw}f$ is in $C^1(\RR,\H)$ if and only if $T_{\omega(w,\placeholder)}f(z) = w \cdot \nabla f(z)-\bar w \cdot z\,f(z) \in \H$.
    For such $f$, the derivative of $W_zf$ at $z=0$ in direction $w$ is given by $\frac it T_{\omega(w,\placeholder)}f$.
\end{lem}


\begin{proof}[Proof of \cref{thm:schwartz_elts}]
    We only prove $n=1$ for the sake of better readability.

    \ref{it:seq} $\Leftrightarrow$ \ref{it:schrodinger}: Put $\psi_m = B e_m$, for $t=1$ this is the $m$-th Hermite function.
    Applying \cite[Thm.\ V.13]{reed1972methods} finishes this step.

    \ref{it:pol} $\Rightarrow$ \ref{it:seq}: Define $N = \frac12 T_{z} T_{\bar z}$ which acts by $Nf(z) = tzf'(z)$.
    One has that $N e_m =  m e_m$.
    Thus, \ref{it:seq} is equivalent to $f \in D(N^\infty) \coloneqq  \bigcap_{k\geq1} D(N^k)$.
    The assumption of \ref{it:pol} implies $(z \partial)^m f = t^{-m} N^{m} f \in \H$ for all $m$, i.e., $f \in D(N^\infty)$.

    \ref{it:schrodinger} $\Leftrightarrow$ \ref{it:bargmann} was proved by Bargmann in \cite{bargmann1967hilbert}.

    \ref{it:seq} $\Rightarrow$ \ref{it:pol}:
    We use the expansion in terms of the monomial basis
    \begin{align*}
        \norm{z^a \partial^b f}^2
        = \sum \abs{\ip f{e_m}}^2 \norm{z^a \partial^b e_m}^2
        &= c \sum \abs{\ip f{e_m}}^2 \prod_{j=0}^{b-1} (m-j) \prod_{k=1}^a (m-b+k)\\
        &\leq c \sum \abs{\ip{f}{e_m}}^2 (m+a)^{a+b} <\infty
    \end{align*}
    for a constant $c= c(a,b,t)$.
    Summability follows from rapid decay of $\ip f{e_m}$.

    \ref{it:FT} $\Rightarrow$ \ref{it:pol}:
    Consider the estimate
    \[
        (2\pi t)^{n}\norm{z^a \partial^bf}
        \leq
        \norm{ z^a \partial^b \big[f(z) e^{-\frac{\abs z^2}{4t}}\big] }_{L^2(\CC^n)} + (4t)^{-b}\norm{z^a \bar z^b f(z) e^{-\frac{\abs z^2}{4t}} }_{L^2(\CC^n)}
        <\infty.
    \]
    Both summands are $L^2$-norms of Schwartz functions on $\CC^n$ and thus finite.

    \ref{it:pol} $\Rightarrow$ \ref{it:FT}: It is clear that $g(z)\coloneqq  (2\pi t)^{-n/2}f(z) e^{-\frac{\abs z^2}{4t}}$ is smooth.
    It suffices to show uniform boundedness $\norm{ z^a \partial^b \bar\partial^c g(z) }_{L^2(\CC^n)}< \infty$ for all $a,b,c \in \NN_0$.
    We can drop the $\bar\partial^c$ because it effectively multiplies by $z^c$.
    Similar to the above, the result follows from
    \[
        (2\pi t)^n\norm{z^a \partial^b g}_{L^2(\CC^n)}
        \leq \norm{z^a \partial^b f} + (4t)^{-b} \norm{z^{a+b} f} <\infty.
    \]

    \ref{it:weyl} $\Rightarrow$ \ref{it:pol}:
    The Wirtinger derivatives of $W_wf$ exist up to any order by assumption.
    By equation \eqref{eq:weyl_exponential} they are given by
    \[
        \partial_w W_w f |_{w=0} = f' = \frac1{2t} T_{\bar z}f
        \quad \text{and} \quad
        \bar\partial_w W_w f |_{w=0} = -\frac1{2t} zf = -\frac1{2t} T_{z}f.
    \]

    \ref{it:pol} $\Rightarrow$ \ref{it:weyl}:
    This follows from \cref{thm:domain_toep}.
    We know that, by assumption, $W_zf$ is differentiable and that the derivative at $0$ in direction $w$ is given by $\frac itT_{\omega(w,\placeholder)}f \in\H$.
    As this function again satisfies condition \ref{it:pol}, the result follows by induction.
\end{proof}

In the Schr\"odinger representation the equivalences \ref{it:schrodinger} $\Leftrightarrow$ \ref{it:seq} and \ref{it:schrodinger} $\Leftrightarrow$ \ref{it:weyl} are well-known (\cite[Thm.\ V.13]{reed1972methods} and \cite[p.\ 827]{howe1980role}, respectively).
The equivalence of \ref{it:FT} and \ref{it:schrodinger} can be deduced from the formalism developed in \cite{keyl2016schwartz} (the proof given here is different).
The characterization \ref{it:bargmann} was shown in \cite{bargmann1967hilbert}.

$\SE$ is a Fr\'echet space with the family of semi-norms $q_{\mu,\nu}(f) = \norm{z^\mu \partial^\nu f(z)}$ indexed by $\mu,\nu\in \NN_0^n$. The Bargmann isometry is a topological isomorphism $B :\SE \to \Sch(\RR^n)$.
Convenient subspaces of $\SE$ are
\begin{equation}\label{eq:PandE}
    \mathcal P = \lh\{e_\nu : \nu\in \NN_0^n\} \equiv \CC[z_1,\dots,\bar z_n] 
   \quad \text{and} \quad
   \mathcal E = \lh\{ K(\placeholder,w) : w\in \CC^n\}. 
\end{equation} 
In fact, $\SE$ is invariant under the whole metaplectic representation. This follows from the fact that the generators of the metaplectic representation are $\frac it$ times Toeplitz operators with real homogenous polynomial symbols of degree two \cite[Ch.\ 4]{folland2016harmonic}. 
$\SE$ is, however, not closed under pointwise multiplication, as $f(z) = e^{\frac1{8t}z^2}$ satisfies \ref{it:FT} but $f(z)^2$ does not.
Note that the set $\SE$ depends on the value of $t$.

We can also make sense of Schwartz elements in the \vectorvalued\ case.
$\SEK$ is the space of $f \in \HK$, such that $\ip{f(\cdot)}{x}  \in \SE$ for all $x\in \K$.
The \vectorvalued\ Schwartz elements are topologized by the ``same'' family of semi-norms: $q_{\mu,\nu}(f) = \norm{z^\mu \partial^\nu f(z)} $.
Since $\SE$ is a nuclear Fr\'echet space, the injective and projective tensor products of $\SE$ and $\K$ coincide and we denote them by $\SE\otimes \K$ \cite[Ch.\ 50]{treves2016topological}.
The tensor product $\SE\otimes\K$ is naturally isomorphic to $\SEK$ with the correspondence determined by $(f\otimes x)(z) = f(z)x$.
\cref{thm:schwartz_elts} holds mutatis mutandis in the \vectorvalued\ case.

The (\vectorvalued) Schwartz elements can be used as a domain for unbounded Toeplitz operators (with \operatorvalued\ symbols).
We have the following

\begin{prop}\label{thm:unb_to}
    Let $f: \CC^n \to \BO\K$ be (weakly) measurable and polynomially bounded, i.e., $\norm{f(z)}  \lesssim 1+\abs z^N$ for some $N>0$.
    Then $f\cdot g\in L^2(\CC^n,\mu;\K)$ for any $g\in\SEK$.
    Thus, the Toeplitz operator $T_fg = P(f\cdot g)$ is well-defined on $\SEK$.
    In fact, $T_f$ leaves $\SEK$ invariant and is a continuous operator $T_f :\SEK \to \SEK$.
\end{prop}

\begin{rem}\label{rem:measurability}
As we assume that the Hilbert space $\K$ is separable, weak measurability of a function $g:\CC^n\rightarrow\BO\K$ is equivalent to its (strong) measurability according to Pettis' Theorem \cite{pettis}. 
To ensure weak measurability of an \operatorvalued\ function $f:\CC^n\rightarrow\BO\K$ is sufficient to prove the weak and, therefore, strong measurability of the point-wise product $f\cdot g:\CC^n\rightarrow\K$ because
\begin{align*}
    \ip{(f\cdot g)(z)}{x} & = \sum_{i\in I}\ip{g(z)}{v_{i}} \ip{f(z)v_{i}}{x},
\end{align*}
is measurable for any $x\in\K$ and some (countable) basis $\{v_{i}\}_{i\in I}$ of $\K$ as the point-wise limit of sums of products of measurable functions.
\end{rem}

\begin{proof}[Proof of \cref{thm:unb_to}]
    The assumption of polynomial boundedness together with \cref{thm:schwartz_elts}, \ref{it:bargmann} immediately implies that $f\cdot g$ is in $L^2(\CC^n,\mu;\K)$.

    We work in the Schr\"odinger representation to show that $T_f$ is a continuous operator on $\SEK$.
    Here $T_f$ acts as the Weyl-pseudodifferential operator with symbol $\widetilde f^{(t /2)}$, which is the heat transform of $f$ after time  $\frac t2$ (cf.\ \cite{berger1994heat}).
    It is clear that $\widetilde f^{(t/2)}$ is smooth and polynomially bounded.
    We pick the order function $m :\CC^n \to [1,\infty), \,z \mapsto (1 +\abs z^2)^{1/2}$.
    Let $N>0$ be such that $\norm{f(z)} \lesssim  m(z)^N$.
    Using Peetre's inequality, we get for $\alpha,\beta\in \NN_0^{n}$
    \begin{align*}
        \norm{\partial^\alpha\bar\partial^\beta\widetilde f^{(t/2)} (z)}
         &\leq (\pi t)^n \int \norm{f(z-w)} \Big| \partial^\alpha_w \bar\partial^\beta_w e^{-\frac{\abs w^2}{4t}}\Big| \d w\\
         &\lesssim   \int m(z-w)^N  \abs w^{\abs\alpha+\abs\beta} e^{-\frac{\abs w^2}{4t}} \d w\\
         &\lesssim m(z)^N \!\int m(w)^{N+\abs\alpha+\abs\beta} e^{-\frac{\abs w^2}{4t}}\d w
         \lesssim m(z)^N.
    \end{align*}
    Proposition A.4 in \cite{teufel2003adiabatic} finishes the proof.
\end{proof}

\begin{rem}
    If one is interested in finding a dense subspace $\mathcal D$ such that for a larger class of symbols, the Toeplitz operators $T_f : \mathcal D \to \mathcal D$ are well-defined,
    then one can modify the construction.
    For example, the approach presented in \cite{bauer2009berezin} allows for symbols with certain exponential growth and the domain $\mathcal D$ used therein is included in $\SE$.
    The invariance of the common domain is necessary for a discussion of products (or commutators) of Toeplitz operators.
    In the extension of the Berger-Coburn estimates in \cref{sec:bce}, we will admit symbols with a significantly weaker growth condition (cf.\ \cref{thm:bce}).
\end{rem}

It is clear that the Weyl-operators $W_z$ leave $\SE$ invariant.
In fact, we can show that the action of the Weyl-operators is strongly smooth on $\SE$.

\begin{lem}\label{thm:schwartz_elts_smooth}
    For any $f\in \SE$ the map $z\mapsto W_z f$ is in $C^\infty(\CC^n,\SE)$.
\end{lem}

\begin{proof}
    The product formula for Toeplitz operators with polynomial symbols (cf.\ \cite[Thm.\ 2]{coburn2001berezin}) implies that the topology of $\SE$ is induced by the equivalent family of semi-norms $q_p(f)\coloneqq  \norm{T_p f}$ where $p$ is a polynomial in $z_1,\dots,z_n,\bar z_1,\dots, \bar z_n$.
    Note that
    \[
        W_w^* T_p W_w = T_{p(z+w)}
        = \sum_{\alpha,\beta \in\NN_0^n}  \frac1{\alpha!\beta!} \,{w^{\alpha}\bar w^\beta}\, T_{\partial^\alpha\bar\partial^\beta p}.
    \]
    Furthermore, we have that $[T_{\omega(w,\placeholder)},T_p] = \frac t{2i} T_{\bar w \cdot \bar\nabla p -w\cdot \nabla p }$.
    Therefore,
    \begin{align*}
        &T_p \Big( \frac1s (W_{sw}-\1) - \frac it T_{\omega(w,z)} \Big)\\
        = & \frac1s \Big( W_{sw} W_{sw}^* T_p W_{sw} - T_p \Big) - \frac it T_{\omega(w,z)} T_p  + \frac it[T_{\omega(w,z)},T_p]\\
        = & \Big( \frac1s (W_w -\1) - \frac itT_{\omega(w,z)} \Big)T_p
        + \text{terms with $s,s^2,\cdots,s^{\text{deg}(p)-1}$}.
    \end{align*}
    Applying this to an element $f\in \SE$ and using \cref{thm:schwartz_elts} proves that
    \[
        q_p\Big(\frac1s(W_{sw}f-f) - \frac it T_{\omega(w,z)}f\Big)
        = \Big\|\frac1s(W_w g-g) -\frac itT_{\omega(w,z)}g\Big\|+ \mathbf{o}(s) \overset{w \to0}{\longrightarrow} 0
    \] with $g = T_p f \in \SE$.
    As the derivatives are again in $\SE$, smoothness follows inductively.
\end{proof}

As anounced earlier, we use $\SE$ as a domain for quadratic forms.
Let $E$ be a topological vector space.
A continuous $E$-valued quadratic form on $\SE$ is a separately continuous and sesquilinear map $A :\SE \times \SE \to E$.
The principle of uniform boundedness \cite[Thm.\ V.7]{reed1972methods} and the fact that $\SE$ is a nuclear Fr\'echet space imply that separate continuity is equivalent to joint continuity. 
We will not always mention the space $E$ and the domain $\SE$ and simply refer to these as ``continuous forms''.

For a continuous form $A$ and a continuous operator $B:\SE\to\SE$ the commutator $[A,B](f,g) \coloneqq  A(Bf,g)-A(f,Bg)$ is well-defined as a continuous form (\vectorvalued\ whenever $A$ is).
We introduce phase space derivatives of continuous forms to be the forms $\partial_j A$ and $\bar \partial_j A$ with
\begin{equation}\label{eq:derivative}
    \partial_j A \coloneqq  -\frac1{2t} [A,T_{\bar z_j} ]
    \quad \text{and} \quad
    \bar \partial_j A \coloneqq  \frac1{2t} [A,T_{z_j} ].
\end{equation}
These are again continuous forms because of \cref{thm:schwartz_elts}.
Higher-order derivatives are defined iteratively $\partial^\beta\bar\partial^\gamma A = (\partial_1)^{\beta_1} \dots (\bar\partial_n)^{\gamma_n}A$, where $\beta,\gamma\in \NN_0^{n}$.
In fact, these derivatives commute with each other because $T_{z_j}$ and $T_{\bar z_j}$ commute with their commutator.

For a continuous form on $\SE$, we define the continuous form $\alpha_zA$ by $\alpha_zA(f,g) \coloneqq  A(W_{-z}f,W_{-z}g)$ and the Berezin transform $\widetilde A (z): \CC^n\to E$ by $\widetilde A(z)\coloneqq  A(k_z,k_z)= \alpha_{-z}A(1,1)$. 
The suggestive definition of phase space derivatives of continuous forms is justified by the next result and the fact that the derivatives of a sufficiently regular function $f$ on phase space coincide with the derivatives of $z\mapsto \alpha_z f$ at $z=0$.

\begin{cor}\label{thm:forms_diff}
    Let $E$ be a Fr\'echet space and let $A$ be an $E$-valued continuous quadratic form on $\SE$.
    Then $z \mapsto \alpha_z A(f,g)$ is a smooth map $\CC^n \to E$ for all $f,g\in\SE$.
    The derivatives are given by
    \[
        \frac{\partial^{\abs\beta+\abs\gamma}}{\partial z^\beta\partial \bar z^\gamma}
        \alpha_z A(f,g) \big|_{z=0} 
        =\partial^\beta\bar\partial^\gamma A(f,g),
        \quad \beta,\gamma\in\NN_0^n.
    \]
    In particular, the Berezin transform $\widetilde A : \CC^n \to E$ is smooth.
\end{cor}

\begin{proof}
    For any continuous semi-norm $q$ on $E$, pick a continuous semi-norm $p$ on $\SE$ such that $q(A(f,g)) \leq p(f)p(g)$ for all $f,g\in\SE$.
    By \cref{thm:schwartz_elts} the commutator $\frac it[A,T_{\omega(w,z)}]$ is again a continuous $E$-valued form on $\SE$.
    Differentiability of $\alpha_z A(f,g)$ now follows from applying the triangle inequality
    \begin{align*}
        &\ q\Big(\frac1s\big(A(W_{sw}f,W_{sw}g) - A(f,g)\big) - \frac it [A,T_{\omega(w,z)}] (f,g)\Big)\\
        \leq&\ q\Big( A\Big( \frac1s(W_{sw}f-f) - \frac it T_{\omega(w,z)}f ,W_{sw}g \Big) \Big)
        + q\Big(A\big(f,\frac1s(W_{sw}g-g) - \frac it T_{\omega(w,z)}g \big)\Big) \\
            &+ \frac1t q\Big( A(T_{\omega(w,z)}f,W_{sw}g- g) \Big).
    \end{align*}
    The first term is bounded by $p\big[ (W_{sw}f-f)/s - (i/t) T_{\omega(w,z)}  f\big] p(W_{tw}g)$ which tends to zero as $s\to0$.
    The second and third term behave similarly.
\end{proof}

If $f: \CC^n \to \BO\K$ is (weakly) measurable and and polynomially bounded, then the Toeplitz operator $T_f : \SEK \to\SEK$ can be viewed as a continuous $\BO\K$-valued quadratic form on $\SE$ by setting 
\[
    \ip{T_f(g,h)x}y \coloneqq  \ip{T_f (g\otimes x)}{h\otimes y} 
\] 
for $x,y\in\K$ and $g,h\in\H$.
We can now look at derivatives of Toeplitz operators. 
We expect that  $\partial^\beta\bar\partial^\gamma T_f = T_{\partial^\beta\bar\partial^\gamma f}$ holds for sufficiently regular $f$ because of the covariance $\alpha_z T_f = T_{\alpha_z f}$.
In fact, the following holds:

\begin{lem}
    Let $f : \CC^n \to \BO\K$ be $k$-times differentiable w.r.t.\ the weak operator topology on $\BO\K$ and assume that $f$ and its derivatives up to $k$-th order are polynomially bounded.
    Then the phase space derivatives $\partial^\beta\bar\partial^\gamma T_f$ up to $k$-th order are again Toeplitz operators on $\SEK$. In fact,
    \[
        \partial^\beta\bar\partial^\gamma T_f = T_{\partial^\beta\bar\partial^\gamma f} \quad \text{if $\abs\beta+\abs\gamma \leq k$}.
    \] 
\end{lem}

\begin{proof}
It is sufficient to consider the case $k= 1$. We use \cref{thm:schwartz_elts_smooth}:
    \begin{align*}
        \partial_jT_f (g,h) = \frac{\partial}{\partial z_j} \ip{T_f W_{-z}g}{W_{-z} h} \big|_{z=0} 
        &= \frac{\partial}{\partial z_j} \ip{T_{\alpha_zf}g}h \big|_{z=0} \\
        &= \frac{\partial}{\partial z_j} \int \ip{f(z+w)g(w)}{h(w)}  \d\mu(w) \big|_{z= 0} \\
        &= \int \ip{(\partial_j f \cdot g)(w)}{h(w)}  \d\mu(w) = \ip{T_{\partial_j f}g}h,
    \end{align*}
    where we used the polynomial boundedness of the derivative and dominated convergence. The same argument works for the $\bar\partial_j$ derivatives.
\end{proof}

Continuous quandratic forms on $\SE$ have been considered in a very similar manner in \cite{keyl2016schwartz} in the context of quantum harmonic analysis \cite{qha} which has recently been used in the framework of the Segal-Bargmann space in \cite{fulsche2020correspondence} (see also \cite{luef2018convolutions,luef2021wiener}).

\section{Berger-Coburn Estimates: \VectorValued\ Case}\label{sec:bce}

This section is about extending the Berger-Coburn estimates \cite{bauer2020berger,berger1994heat} to Toeplitz operators with \operatorvalued\ symbols. 
This is achieved by the following theorem.
For an \operatorvalued\ symbol $f :\CC^n \to \BO\K$ we use the notation $\norm{f}_\infty \coloneqq  \sup_{u\in \CC^n} \norm{f(u)}$.

\begin{thm}\label{thm:bce}
    Let $f: \CC^n \to \BO\K$ be weakly measurable such that $K(\placeholder,w) \norm{f(\placeholder)} \in L^2(\CC^n,\mu)$ for all $w\in\CC^n$.
    Let $s \in (0,\frac t2)$. Then there is a $C(s,t)>0$ (depending on $s,t$ only) such that 
    \begin{equation}\label{eq:bce}
        \norm{T_f} \leq C(s,t) \norm{\widetilde f^{(s)}}_\infty. 
    \end{equation}  
\end{thm}

An off-diagonal extension of the heat transform of a weakly measurable function $f:\CC^n\to\BO\K$ is
\begin{equation}\label{GL_off_diagonal_heat_transform_definition}
    \widetilde{f}^{(t)} (z,w)\coloneqq \int f(u)k_z(u) \overline{k_w(u)} \d\mu(u),
\end{equation}
where $z,w \in \CC^n$. 
Note that (\ref{GL_off_diagonal_heat_transform_definition}) defines an $\BO\K$-valued function on $\CC^n \times \CC^n$ and that the heat transform is given by $\widetilde{f}^{(t)} (z)= \widetilde{f}^{(t)} (z,z)$ for all $z \in \CC^n$. 
A direct calculation shows the norm estimate 
\[
    \norm{\widetilde{f}^{(t)}(w,z)} \leq e^{\frac{-|z|^2}{4t}} \bigg{(} \int \|f(u)\|^2 
    |K(u,z)|^2 \d\mu(u)\bigg{)}^{\frac{1}{2}}.
\]
The next lemma shows a useful integral representation of $T_f$. 

\begin{lem}\label{Lemma_TO_integral_representation_heat_transform}
    Let $f: \mathbb{C}^n \to \BO\K$ be as in \cref{thm:bce} and assume that $g \in \HK$ fulfills
    \[
        \|g(\placeholder)\|e^{\frac{|\placeholder|^2}{4t}} \in L^1(\mathbb{C}^n, \mu). 
    \]
    With $z \in \mathbb{C}^n$ we have 
    \begin{equation*}
        T_fg(z)= \int e^{ \frac{1}{4t} |z-w|^2+ \frac{1}{2t} \Re(z \cdot \overline{w})} \widetilde{f}^{(t)}(w,z)g(w) \d\mu(w). 
    \end{equation*}
\end{lem}

\begin{proof} 
    Note that the integral on the right-hand side exists under the assumptions on $f$ and $g$ and the above estimate.
    Let $x \in \K$ and $z \in \CC^n$ be arbitrary.  Then \eqref{eq:toeplitz_kernel} implies 
    \begin{equation*}\tag{+}\label{Expression_1}
        \big{\langle} T_fg(z), x \big{\rangle}  
        = \int \big{\langle} f\cdot g(w) , x \big{\rangle}  K(z,w)  \d\mu(w).
    \end{equation*}
    On the other hand 
    \begin{align*}
        \big{\langle} \widetilde{f}^{(t)}(w,z) g(w), x \big{\rangle} 
        &=\int \big{\langle} f(u) g(w), x \big{\rangle}  k_w(u) \overline{k_z(u)}  \d\mu(u) \\
        &=e^{-\frac{|z|^2+|w|^2}{4t}} \int \big{\langle} g(w), f(u)^*x \big{\rangle}  K(u,w) K(z,u)  \d\mu(u). 
    \end{align*}
    Applying the reproducing kernel formula shows 
    \begin{align*}
        \int e^{\frac{1}{4t} |z-w|^2+ \frac{1}{2t} \Re(z \cdot \overline{w})} &\big{\langle} \widetilde{f}^{(t)}(w,z) g(w), x \big{\rangle}  \d\mu(w) \\
        =&\iint\big{\langle} g(w), f(u)^*x \big{\rangle}  K(u,w) K(z,u)  \d\mu(u)\d\mu(w) \\
        =&\int K(z,u) \big{\langle} (f\cdot g)(u),x \big{\rangle}   \d\mu(u). 
    \end{align*}
    According to our assumptions on $f$ and $g$  the above double integral converges absolutely which justifies the change of order in integration.
    Since $x \in \K$ was chosen arbitrarily we obtain the result by comparing with \eqref{Expression_1}. 
\end{proof}
Now, we can prove the \operatorvalued\ version of \cite[Lem.\ 3.2]{bauer2020berger}, which in the case $z=w$ reduces to the semi-group property of the heat transform: 

\begin{lem}\label{Lemma_extension_semi_group_property}
    Let $z,w \in \CC^n$ and $0 < s < t$. With $f: \CC^n \rightarrow \BO\K$ as above we have 
    \begin{equation*}
        \widetilde{f}^{(t)}(w,z)= e^{\frac{s}{4t(t-s)}{|w-z|^2}- \frac{is}{2t(t-s)} \Im(z \cdot \overline{w})} \bigl(\widetilde f^{(s)} \widetilde{\bigr)\,}^{(t-s)} (w,z).
    \end{equation*}
\end{lem}

\begin{proof}
    We reduce the statement to the scalar-valued case (see \cite{bauer2020berger}). Let $x, y \in \K$ be arbitrary and 
    define 
    \[
    f_{x,y}:\CC^n \rightarrow \CC: f_{x,y}(z)\coloneqq \langle f(z)x, y \rangle.
    \]
    We need the (scalar-valued) Segal-Bargmann space $\H_{t-s}$ with quantization parameter $(t-s)$. We denote the inner product of this space by $\ip\placeholder\placeholder_{t-s}$ and the normalized reproducing kernels by $k_z^{t-s} \in \H_{t-s}$ with $z\in\CC^n$.
    According to \cite[Lemma 3.2]{bauer2020berger} we obtain 
    \begin{align*}
        \big{\langle} \widetilde{f}^{(t)}(w,z)x, y\big{\rangle} 
        = \widetilde{f_{x,y}}^{(t)}(w,z)
        =e^{\frac{s}{4t(t-s)}{|w-z|^2}- \frac{is}{2t(t-s)} \Im(z \cdot \overline{w})} \big{\langle} \widetilde{f_{x,y}}^{(s)} \cdot k_w^{t-s}, k_z^{t-s} \big{\rangle}_{t-s}. 
    \end{align*}
    A straightforward calculation shows that 
    \begin{equation*}
        \big{\langle} \widetilde{f_{x,y}}^{(s)} \cdot k_w^{t-s}, k_z^{t-s} \big{\rangle}_{t-s}
        = \Big{\langle}  \bigl(\widetilde f^{(s)} \widetilde{\bigr)\,}^{(t-s)} (w,z)\,x, y \Big{\rangle} . 
    \end{equation*}
    This finishes the proof, since $x,y\in\K$ were arbitrary.
\end{proof}
We need the following estimate: 
\begin{lem}\label{lemma_off_diagonal_estimate}
    Let $f: \CC^n \rightarrow \BO\K$ be uniformly bounded, then 
    \begin{equation}\label{lemma_off_diagonal_estimate_GL1}
        \big{\|} \widetilde{f}^{(t)}(w,z) \big{\|} \leq \norm{f}_\infty \: e^{- \frac{1}{8t} |w-z|^2}. 
    \end{equation}
\end{lem}

\begin{proof}
    We estimate the norm of the \operatorvalued\ integral as follows 
    \begin{align*}
        \Big{\|} \int f(u) k_w(u) \overline{k_z(u)}  \d\mu(u) \Big{\|} \leq e^{- \frac{|z|^2+|w|^2}{4t}}  \norm f_\infty \int \Big{|} e^{\frac{u \cdot \overline{w}}{2t}+ \frac{\overline{u} \cdot z}{2t}} \Big{|}  \d\mu(u). 
    \end{align*}
    The integral on the right hand side can be calculated explicitly and  one obtains (\ref{lemma_off_diagonal_estimate_GL1}).
\end{proof}
Using the above lemmas we can prove \cref{thm:bce}:

\begin{proof}[Proof of \cref{thm:bce}]
    The space of $g\in \HK$ that satisfy the assumption in \cref{lemma_off_diagonal_estimate} is dense, hence it suffices to show $\norm{T_f g} \leq C(s,t) \norm{\tilde f^{(s)}}_\infty \norm{g}$ for such $g$.
    From the integral representation in \cref{Lemma_TO_integral_representation_heat_transform} and standard estimates, one has for $z\in \CC^n$ 
    \begin{align*}
        \big{\|} T_fg(z)\big{\|}  \leq \int e^{ \frac{1}{4t} |z-w|^2 + \frac{1}{2t} \Re(z \cdot \overline{w})} \Big{\|} \widetilde{f}^{(t)}(w,z) g(w) \Big{\|}   \d\mu(w). 
    \end{align*}
    \cref{Lemma_extension_semi_group_property} and  \cref{lemma_off_diagonal_estimate} with $f: \CC^n \rightarrow \BO\K$ replaced 
    by $\widetilde{f}^{(s)}$ show for any $z,w \in \CC^n$ 
    \begin{align*}
        \big{\|} \widetilde{f}^{(t)}(w,z) g(w) \big{\|}  
        &=e^{\frac{s}{4t(t-s)}|z-w|^2}\Big{\|} \big{(} \widetilde{f}^{(s)}\widetilde{\big{)}}^{ (t-s)}(w,z) g(w) \Big{\|} \\
        & \leq e^{\big{(}\frac{s}{4t(t-s)}- \frac{1}{8(t-s)} \big{)}|z-w|^2}   \| \widetilde{f}^{(s)}\|_{\infty} \|g(w)\| . 
    \end{align*}
    Inserting this estimate into the above integral shows 
    \begin{equation*}
        \big{\|} T_fg(z)\big{\|}  \leq  \| \widetilde{f}^{(s)}\|_{\infty} \int e^{\frac{1}{8(t-s)} |z-w|^2+ \frac{1}{2t} \Re(z \cdot \overline{w})} 
        \|g(w)\|   \d\mu(w). 
    \end{equation*}
    We now estimate the integral on the right-hand side by applying the Cauchy-Schwarz inequality. Define the symmetric kernel
    \begin{equation*}
        L(z,w)\coloneqq  e^{\frac{|z-w|^2}{8(t-s)}+ \frac{1}{2t} \Re(z \cdot \overline{w})} >0. 
    \end{equation*}
    We obtain 
    \begin{multline}\tag{$*$}\label{GL_berger1994heat_Theorem_estimate}
        \int L(z,w) \| g(w)\|   \d\mu(w)\\
        \leq \left(\int L(z,w) e^{\frac{|w|^2}{4t}}  \d\mu(w) \right)^{\frac{1}{2}} 
        \left( \int \| g(w) \|^2  L(z,w) e^{- \frac{|w|^2}{4t}}  \d\mu(w) \right)^{\frac{1}{2}}. 
    \end{multline}
    From $t-s>\frac{t}{2}$ we conclude that there is $\varepsilon >0$ depending on $s$ such that 
    \begin{equation*}
        \frac{1}{8(t-s)} \leq \frac{1}{4t}- \varepsilon. 
    \end{equation*}
    Therefore, we have for all $z,w \in \CC^n$  
    \begin{equation*}
        \frac{|z-w|^2}{8(t-s)} + \frac{\Re(z \cdot \overline{w})}{2t} \leq \frac{|z|^2}{4t} + \frac{|w|^2}{4t} - \varepsilon |z-w|^2,
    \end{equation*}
    implying the pointwise estimate 
    \begin{equation*}\tag{$**$}\label{GL_estimate_kernel_L}
        0 < L(z,w) \leq e^{\frac{|z|^2+|w|^2}{4t}- \varepsilon |z-w|^2}. 
    \end{equation*}
    \par 
    We can now bound the first integral on the right hand side in (\ref{GL_berger1994heat_Theorem_estimate}) 
    \begin{equation*}
        \int L(z,w) e^{\frac{|w|^2}{4t}}  \d\mu(w) \leq C_{\varepsilon,t} e^{\frac{|z|^2}{4t}} \hspace{3ex} \text{with} \hspace{3ex} C_{\varepsilon,t}
        \coloneqq  \frac{1}{(2\pi t)^{n}}  \int e^{-\varepsilon |w|^2} \d w .  
    \end{equation*}
    Therefore
    \begin{equation*}
        \|T_fg(z)\|^2  \leq C_{\varepsilon,t}  \| \widetilde{f}^{(s)}\|^2_{\infty}  e^{\frac{|z|^2}{4t}} \int \| g(w) \|^2  L(z,w) e^{- \frac{|w|^2}{4t}}  \d\mu(w). 
    \end{equation*}
    Integration of this inequality, Fubini's theorem and a second application of the estimate (\ref{GL_estimate_kernel_L}) show 
    \begin{align*}
        \|T_fg\|^2
        &= \int \|T_fg(z)\| ^2  \d\mu(z)\\
        &\leq \frac{C_{\varepsilon,t}}{(2\pi t)^n} \| \widetilde{f}^{(s)}\|_{\infty}^2
        \int \|g(w)\|^2  \int L(z,w) e^{- \frac{|z|^2+|w|^2}{4t}} \d z  \d\mu(w)\\ 
        & \leq C_{\varepsilon,t}^2  \| \widetilde{f}^{(s)}\|^2_{\infty}  \norm{g}^2.
    \end{align*}
    Hence, we obtain $\|T_f\| \leq C_{\varepsilon,t} \| \widetilde{f}^{(s)}\|_{\infty}$ as required. Note that $C_{\varepsilon,t} \rightarrow \infty$ as $s \uparrow \frac{t}{2}$. 
\end{proof}


\section{Self-Adjointness}\label{sec:esa}

Recall that we denote the heat transform at time $s$ of a polynomially bounded function $f$ on $\CC^n$ by $\widetilde f^{(s)}$, i.e.,
\begin{equation}\label{eq:heattransform}
    \widetilde f^{(s)}(z) =  (2\pi s)^{-n/2}\!\int e^{-\frac{\abs w^2}{2s}} f(z+w) \d w.
\end{equation} 
This makes sense as a Bochner-integral if $f$ is $\BO\K$-valued, and we set $\widetilde f^{(0)}= f$.
We denote by $\BO\K_h$ the (real) Banach space of hermitian bounded operators.

Let $(E,\norm\placeholder)$ be some normed space. A function $f : \CC^n \to E$ has \emph{bounded oscillation} \cite{bauer2015heat}, if it is continuous and satisfies
\begin{equation}\label{eq:oscillation}
    \norm{f(z)-f(w)}  \lesssim 1+ \abs{z-w},\quad \text{for all $z,w\in\CC^n$}.  
\end{equation} 
This can be seen to be equivalent to $\sup\{ \norm{f(z) -f(w)} : z,w \in\CC^n, \abs w <1\} <\infty$ (cf.\ \cite{bauer2005mean}).

\begin{thm}\label{thm:main}
    Let $f :\CC^n \to \BO{\K}_h$ be a weakly measurable polynomially bounded symbol such that there is a time $s\in [0,\frac t2)$ for which the first-order derivatives of $\widetilde f^{(s)}$ have bounded oscillation. 
    Then $T_f$ is self-adjoint on its natural domain
    \begin{equation}\label{eq:toeplitz_domain}
        D(T_f) 
        = \bigl\{ g\in\HK : \text{the integral \eqref{eq:toeplitz_kernel} exists $\forall z\in \CC^n$ and $T_fg \in\HK$}\bigr\}.
    \end{equation} 
    The algebraic tensor products $\mathcal P\otimes \K$, $\mathcal E\otimes\K$, and thus $\SEK$, are cores for $T_f$.
\end{thm}

If $s=0$ then the symbol $f$ has to be continuously differentiable w.r.t.\ the weak operator topology on $\BO\K$.
The spaces $\mathcal P$, and $\mathcal E$ are defined in \eqref{eq:PandE}.
We divide the proof in two steps: 
In the first step we treat the case $s=0$ by use of the commutator theorem, and in the second step we use the \operatorvalued\ Berger-Coburn estimate in \cref{thm:bce} to reduce the statement to the case $s=0$.

The main tool in the proof of \cref{thm:main} is the \emph{commutator theorem} (cf.\ \cite[Thm.\ X.37]{reed1975methods}).

\begin{lem}[Commutator Theorem]\label{thm:commutator}
    Let $N\geq c\,\1$, $c>0$, be a self-adjoint operator on a separable Hilbert space $\mathcal V$ with core $\mathcal D$.
    Let $A : \mathcal D \to \mathcal V$ be a symmetric operator, such that
    \begin{enumerate}[(i)]
        \item 
            $\norm{Ax} \lesssim \norm{Nx}$ for all $x\in \mathcal D$. 
        \item 
            $\abs{\ip{Nx}{Ax}-\ip{Ax}{Nx}} \lesssim \ip{Nx}{x}$ for all $x\in \mathcal D$.
    \end{enumerate} 
    Then $A$ is essentially self-adjoint, $D(N) \subset D(\bar{A})$ and any core for $N$ is also a core for the self-adjoint extension $\bar A$.
\end{lem}

\begin{proof}[Proof of \cref{thm:main}. Step 1]
    We prove the case $s=0$ under the assumption that $f$ is differentiable w.r.t.\ the weak operator topology on $\BO\K$.

    We will apply \cref{thm:commutator} with $N= N_1 \otimes \1_\K$ where $N_1= \frac12 T_{\abs z^2}$,
    $\mathcal D\in \{ \mathcal P\otimes \K , \mathcal E\otimes \K , \SEK\}$ and $A = T_f :\mathcal D \to \H$.
    With the composition formula for Toeplitz operators with polynomial symbols (cf.\ \cite[Thm.\ 2]{coburn2001berezin}) it can be checked that
    \[
        N_1 =  \frac12 T_z T_{ \bar z}+ nt\1_\H \quad 
        \text{ and } \quad
        N_1 e_\gamma = t\Big(\abs\gamma + n \Big) e_\gamma.
    \] 
    Thus, $N_1 \geq {nt} \1$ is self-adjoint on $D(N_1) =  \{f\in \H : \sum_\gamma \abs{\ip{f}{e_\gamma}}^2 \abs\gamma < \infty\}$.
    The operator $N_1$ is known as the quantum harmonic oscillator.
    For the proof we need the following easy observations:
    \begin{enumerate}[(a)]
        \item\label{it:positivity}
            $\ip{T_hg}g \geq 0$ for all $g\in \SEK $ and all polynomially bounded positive \operatorvalued\ $h :\CC^n\to \BO\K$.
        \item \label{it:kadison_schwartz}
            $\ip{T_h^2g}g \leq  \ip{T_{h^2}g}g$ for all $g\in \SEK$ and all polynomially bounded symbols $h :\CC^n\to \BO\K_h$.
        \item \label{it:strict_positivity}
            For any $a_0,  \dots, a_m \in \RR$, $m\in \NN$, there is a $c>0$ such that $\sum_{k= 0}^m a_k N^k \leq c N^m$.
    \end{enumerate} 

    As is well known, the harmonic oscillator is one of the infinitesimal generators of the metaplectic representation (cf. \cite[Chap.\ 4]{folland2016harmonic}).
    Namely, it corresponds to the classical Hamiltonian function $N_{\text{cl}}(x,\xi ) = \frac12( \abs x^2 + \abs \xi^2)$.
    In the complex phase space setting, the symplectic flow that is generated by $N_{\text{cl}}$ is $z \mapsto e^{i\theta}z$, where $\theta\in\RR$ is the time parameter.
    The unitary one-parameter group generated by $N$ is
    \[
        e^{i\theta N /t} g(z) = e^{in\theta} g(e^{i \theta}z),
        \quad g\in \HK. \tag{++}\label{eq:fractalFT}
    \] 
    Obviously, the dense domain $\mathcal D$ is invariant under $e^{i\theta N/t}$.
    From the core theorem it follows that $\mathcal D$ is a core for $N$ (cf.\ \cite[Prop.\ II.1.7]{engelnagel}).

    We start by checking condition (i) of the commutator theorem.
    The assumptions guarantee that the first-order derivatives of $f$ are linearly bounded, which in turn implies that $f$ is quadratically bounded, 
    i.e., $\norm{f(z)} \lesssim (1+\abs z^2)$ for all $z$.
    Applying \ref{it:kadison_schwartz}, we estimate
    \[
        \norm{T_f g}^2 = \ip{T_f^2 g}g \leq \ip{T_{f^2}g}g
        \leq \ip{(\1_\H+ 2T_{\abs z^2} + T_{\abs z^4})\otimes\1_\K g}g.
    \]
    The Toeplitz composition formula \cite{bauer2009berezin} shows that $T_{\abs z^4} = (T_{\abs z^2})^2 + 2t T_{\abs z^2} - 4n t^2 \1$ which is a polynomial in $N$.
    Inserting this into the above equation and using \ref{it:strict_positivity} shows that
    $ \norm{T_fg} \lesssim \ip{N^2g}g^{1 /2} = \norm{Ng}$. 

    To check the second condition of the commutator theorem, we use \eqref{eq:fractalFT} to get $ e^{i\theta N /t}T_f e^{-i\theta N /t} = T_{f(e^{i\theta} \placeholder)}$. 
    This implies
    \begin{multline*}
        \frac it \ip{[N,T_f]g}g 
        = \frac{\mathrm d}{\mathrm d\theta}\ip{e^{i\theta N /t}T_f e^{-i\theta N/t} g}g \Big|_{\theta=0}  
        = \frac{\mathrm d}{\d\theta} \int \ip{f(e^{i\theta}z)g(z)}{g(z)}  \d \mu(z) \Big|_{\theta=0}, \\
    \end{multline*}
    where $\partial_\theta f(z) = \frac{\mathrm d}{\mathrm d\theta}f(e^{i\theta}z)\big|_{\theta=0} = i(z \cdot \bar \nabla f(z) - \bar z \cdot \nabla f(z)) \in \RR$.
    For the last equality we used dominated convergence to exchange differentiation with the integral sign.
    By assumption the derivatives of $f$ have bounded oscillation.
    Thus, $ \norm{\partial_\theta f(z)} \lesssim (1+\abs z ^2) $.
    Applying \ref{it:positivity} and \ref{it:strict_positivity} yields
    \[
        {\ip{[N,T_f]g}g}
        \lesssim \ip{(\1+T_{\abs z^2})\otimes \1_\K g}g
        \lesssim \ip{Ng}g.
    \] 
    Applying the same argument to $-f$ proves (ii).

    We now denote by $T_f$ the unique self-adjoint extension of $T_f :\mathcal D \to \H$ which is independent of the choice of $\mathcal D\in\{\mathcal P \otimes \K,\mathcal E\otimes\K,\SEK\}$.
    The domain $D(T_f)$ is just the domain of the adjoint operator $(T_f |_{\mathcal{E}\otimes\K})^*$ which is easily checked to be characterized as claimed (cf.\ \cite[Prop.\ 1.4]{janas1}).
\end{proof}

For the second step we need the following

\begin{lem}\label{thm:taylor}
    Let $(E,\norm\placeholder)$ be a Banach space and $F :\RR^n \to E$ be a differentiable function such that its derivatives have bounded oscillation.
    Then the function 
    \[
        R_x(y) =  F(y+x) - F(y) - x\cdot \nabla F(y), \quad x,y\in\RR^n
    \] 
    is uniformly bounded with $\norm{R_x}_\infty\leq c(\abs x+\abs x^2)$ 
    for any $c>0$ such that $\norm{\partial_jF(x)-\partial_jF(y)} \leq c(1+\abs{x-y})$ for all $j= 1,\dots,n$.
\end{lem}

The case $n=1$ follows directly from the fundamental theorem of calculus. 
The case $n>1$ can be reduced to the one-dimensional situation by considering all one-dimensional linear submanifolds of $\RR^n$ and by noting that the constants discussed in the assertion work for all of these simultaneously, which proves the claim.


\begin{proof}[Proof of \cref{thm:main}. Step 2]
    This part of the proof is based on the Berger-Coburn estimate \cite{berger1994heat}, which yields
    \[
        \norm{T_f -T_{\widetilde f^{(s)}}} \leq c(s) \norm{\widetilde f^{(s)}-\widetilde f^{(2s)}}_\infty, 
        \tag{+++} \label{eq:applied_bce}
    \]  
    where $\norm{f}_\infty = \sup\{ \norm{f(z)} : z\in\CC^n\}$ for $f:\CC^n \to\BO\K$.
    In \cref{sec:bce} it was shown that this estimate also makes sense in the \vectorvalued\ case.

    As in \cref{thm:taylor}, we define $R_w(z) = \widetilde f^{(s)}(w+z)-\widetilde f^{(s)}(z) - w \cdot \nabla \widetilde f^{(s)}(z) - \bar w \cdot \bar\nabla \widetilde f^{(s)}(z)$. 
    We estimate the right-hand side of \eqref{eq:applied_bce}:
    \begin{align*}
        \norm{\widetilde f^{(s)}(z)-\widetilde f^{(2s)}(z)}_\infty
        &= (2\pi s)^{-n} \Big\|\int \!e^{-\frac{\abs w^2}{2s}}(\widetilde f^{(s)}(z)- \widetilde f^{(s)}(z+w) ) \d w \Big\|_\infty\\
        &=  (2\pi s)^{-n} \Big\|\int e^{-\frac{\abs w^2}{2s}} \Big(w\cdot \nabla \widetilde f^{(s)}(z) + \bar w \cdot \bar\nabla \widetilde f^{(s)}(z)+ R_w(z) \Big) \d w \Big\|_\infty \\
        &\lesssim \int (\abs w+\abs w^2) \,e^{-\frac{\abs{w}^2}{2s}} \d w
        < \infty.
    \end{align*}
    We used that odd moments of a normal distribution vanish together with \cref{thm:taylor}.
    This proves that $T_f$ is a bounded (and hermitian) perturbation of $T_{\widetilde f^{(s)}}$. The latter is essentially self-adjoint since we assume the first derivatives to have bounded oscillation and can apply the case $s=0$ of the theorem, which has been proved in step one.
\end{proof}

\begin{rem}\label{rem:janas}
    We want to discuss the assumption of bounded oscillation.
    Let $E$ be a normed space and let $b :\RR_+ \to\RR_+$ be any continuous increasing function.
    Let $f : \CC^n \to E$ be a function such that
    \[
        \norm{f(z)-f(w)} \lesssim b(\abs{z-w})
    \]
    for all $z,w$. Then the same is true for $b(r)=1+r$, i.e., $f$ has bounded oscillation provided that $f$ is continuous.
    To see this, put $a = z-w$ and let $N = \min (\NN \cap [\abs a,\infty))$.
    Since $f(z)-f(w) = \sum_{k=0}^{N-1} f(w+(k/N)a)-f(z+((k+1)/N)a)$ we get that $\norm{f(z)-f(w)} \leq N b(\abs{a} /N) \leq b(1)(1+\abs{z-w})$.

    This applies, for example, to \cite[Thm.\ 3.3]{janas3} which states that $T_f$ is self-adjoint if $\norm{f(z)-f(w)} \lesssim \exp(\abs{z-w}^2/8t)$, and we learn that this is the case precisely if $f$ satisfies the oscillation bound \eqref{eq:oscillation} (which is independent of $t$).
    Using the Berger-Coburn estimate as in step two of the proof of \cref{thm:main}, it follows that the oscillation bound needs only hold for the heat transformed symbol $\widetilde f^{(s)}$ for time $s \in [0,\frac t2)$.
    Standard properties of the heat transform now show that \cref{thm:main} is a generalization of \cite[Thm.\ 3.3]{janas3}.
\end{rem}

\begin{cor}\label{thm:magnetic}
    Let $f :\CC^n \to \BO\K_h$ be continuously differentiable w.r.t.\ the weak operator topology. 
    We denote by $\partial_\theta f$ the function
    \[
        \partial_\theta f(z) \coloneqq   z\cdot \bar\nabla f(z) - \bar z\cdot \nabla f(z)
        =\frac{\mathrm d}{\mathrm{d}\theta} f(e^{i\theta}z) \big|_{\theta=0}. 
    \]
    If $\norm{f(z)}$ and $\norm{\partial_\theta f(z)}$ are quadratically bounded functions, then the Toeplitz operator $T_f$ is self-adjoint on the domain \eqref{eq:toeplitz_domain} and the algebraic tensor products $\mathcal P \otimes\K$ and $\mathcal E \otimes\K$ as well as $\SEK$ are cores for $T_f$.
\end{cor}

\begin{proof}
    Step one of the proof of \cref{thm:main} actually only uses these weaker assumptions and thus shows the result.    
\end{proof}

This criterion applies, for example, to hermitian product functions $f(z) = f_1(\Re z)f_2(\Im z)$ of symbols $f_i \in C^1(\RR^n,\BO\K)$ with bounded derivatives (see \cref{exa:magnetic}).

We can generalize \cref{thm:main} to a class of continuous hermitian quadratic forms on $\SE$.
The idea is to show that for sufficiently regular forms $A$ the Berezin transform $\widetilde A$ is differentiable and that $A$ is a bounded perturbation of the operator $T_{\widetilde A}$ to which we may apply \cref{thm:main}.

Let $A$ be a continuous $\BO{\K}$-valued quadratic form on $\SE$, i.e., $A:\SE\times\SE \to \BO{\K}$ is sesquilinear and jointly continuous in the topology of $\SE$.
We call $A$ hermitian if for all $f \in \SE$ the operator $A(f,f)$ is self-adjoint.
As usual $\norm{A}$ denotes the infimum of all $C \in (0,\infty]$ with $\norm{A(f,g)} \leq C \norm f \norm g$ for all $f,g \in \SE$,
and the form $A$ is said to be bounded if $\norm A<\infty$.
If $A$ is bounded there is a unique $\hat A\in \BO{\HK}$, 
such that $\ip{A(f,g)x}y = \ip{\hat A f\otimes x}{g\otimes y}$ for all $f,g \in \SE$ and $x,y\in\K$.

In terms of the phase space translations $\alpha_zf= f(\placeholder +z)$, we can express bounded oscillation of a function $f$ on $\CC^n$ as $\norm{\alpha_z f -f}_\infty \lesssim 1+\abs z$.
Similarly, we say that a quadratic form $A$ on $\SE$ has bounded oscillation if $\alpha_z A-A$ is bounded with $\norm{\alpha_z A-A} \lesssim 1+\abs z$. 

\begin{thm}\label{thm:forms}
    Let $A$ be a hermitian $\BO\K$-valued continuous quadratic form on $\SE$ with first derivatives that satisfy the oscillation bound 
    \begin{equation}\label{eq:oscillation_forms}
        \norm{\alpha_z \partial_jA - \partial_j A} \lesssim 1+\abs z
        \quad\text{and}\quad 
        \norm{\alpha_z \bar\partial_jA - \bar\partial_j A} \lesssim 1+\abs z
    \end{equation} 
    for all $z\in\CC^n$. In particular, the forms $\alpha_z \partial_j A-\partial_j A$ and $\alpha_z \bar\partial_j A-\bar\partial_j A$ are bounded for all $j=1,\dots,n$.
    Then there is an essentially self-adjoint operator $\hat A :\SEK \to \HK$ with 
    $\ip{A(f,g)x}y = \ip{\hat A f\otimes x}{g\otimes y}$ for all $f,g \in \SE$ and $x,y\in\K$.
\end{thm}

\begin{proof}
    We can express the quadratic form $A - T_{\widetilde A}$ as an integral
    \begin{equation*}\label{eq:difference_to_regularization}
        A(f,g)- \ip{T_{\widetilde A} f}g
        = \int \big[ A - \alpha_w A \big] (f,g) \d\mu(w)
        \quad \text{ for $f,g \in \SE$.} \tag{++}
    \end{equation*}
    We will show that \eqref{eq:difference_to_regularization} is a bounded quadratic form on $\SE$ and that the first-order derivatives derivatives of the Berezin transform $\widetilde A$ have bounded oscillation.
    This finishes the proof because it follows that $A$ is a bounded perturbation of $T_{\widetilde A}$ which is self-adjoint by \cref{thm:main} and \cref{thm:taylor}.

    For fixed $f,g \in\SE$, put $F(z) = \alpha_z A(f,g) \in C^\infty(\CC^n,\BO\K)$.
    It follows that 
    \[
        \frac{\partial^{\abs\beta+\abs\gamma}}{\partial z^\beta\partial\bar z^\gamma} 
        F(z) \big|_{z= 0} 
        = \partial^\beta\bar\partial^\gamma A (f,g).
    \]
    Thus, the assumption implies that the functions $\partial^\beta\bar\partial^\gamma F$ with $\abs\beta+\abs\gamma=1$, and hence also $\partial_j\widetilde A$ as well as $\bar\partial_j\widetilde A$, have bounded oscillation.
    We define $R_z$ as in \cref{thm:taylor} with $F(z) = \alpha_z A(f,g)$, the lemma implies that the remainder is bounded by $\norm{R_{z}}_\infty \lesssim \norm f \norm g(\abs z +\abs z^2)$.
    Consequently, \eqref{eq:difference_to_regularization} yields 
    \begin{align*}
        \norm{ A(f,g)- \ip{T_{\widetilde A} f}g}
        &= \Big\|\int \bigg( w\cdot \nabla F(0)+\bar w \cdot \bar\nabla F(0) + R_{w}(0) \bigg) \d\mu(w)\Big\| \\
        &= \Big\|0 + \int R_{w}(0) \d\mu(w) \Big\|
        \lesssim \norm f\norm g \int (\abs w+\abs w^2) \d\mu(w) 
    \end{align*}
    for a suitable constant $c>0$. 
    We used the triangle inequality, the estimate on the norm of the remainder and the fact that all odd moments of a normal distribution vanish.
\end{proof}

The theorem applies, in particular, if the second derivatives of $A$, i.e., the forms $\partial^\beta\bar\partial^\gamma A$ for $\beta,\gamma\in\NN_0^n$ with $\abs\beta+\abs\gamma = 2$, are already bounded.

Of course, the game of ``phase space derivatives of quadratic forms'' can also be played in the Schr\"odinger representation. 
Here, one thinks of the phase space as the symplectic vector space $\RR^{n}\times \RR^n \equiv \CC^n$.  
For simplicity we set $t\equiv\hbar=1$.
Denote by $Q_j$ and $P_j$ the $j$-th position and momentum operator, respectively, i.e., $Q_j\psi (x)= x_j \psi(x)$ and $P_j \psi(x) = -i \partial_j \psi(x)$.
The Weyl-operators are given by $W_{x,\xi } = \exp\{i\sum_{j= 1}^n(\xi_j Q_j - x_j P_j)\}$, where $(x,\xi ) \in \RR^n\times\RR^n$, and define a strongly continuous representation of the Heisenberg group $\mathbb H_n$.
As before we put $\alpha_{x,\xi }A(\psi,\varphi)= A(W_{x,\xi }^*f,W_{x,\xi }^*f)$.
The phase space derivatives of continuous forms on $\Sch(\RR^n)$ are given by $\partial_{j}A = -i[P_j,A]$ and $\partial_{j+n}A = i[Q_j,A]$ for $j=1,\dots,n$. 
We get the following version of \cref{thm:forms} in the Schr\"odinger representation:

\begin{thm}\label{thm:forms_schrodinger}
    Let $A$ be a hermitian $\BO\K$-valued continuous quadratic form on $\Sch(\RR^n)$ with derivatives that satisfy the oscillation bound
    \[
        \norm{\alpha_{x,\xi }\partial_j A-\partial_j A} \lesssim 1+\abs x + \abs \xi
    \] 
    for all $x,\xi \in\RR^n$ and $j= 1,\dots,2n$. In particular, the forms $\alpha_{x,\xi}\partial_j A-\partial_jA$ are bounded.
    Then there is an essentially self-adjoint operator $\hat A :\Sch(\RR^n;\K)\to L^2(\RR^n;\K)$ such that $\ip{A(\psi,\varphi)x}y = \ip{\hat A \psi\otimes x}{\varphi\otimes y}$ for all $\psi,\varphi \in L^2(\RR^n)$, $x,y\in\K$.
\end{thm}

\section{Examples and Applications}\label{sec:examples}

We collect some examples and applications of our results.
In \cref{thm:schroedinger,exa:magnetic} we prove self-adjointness of physically relevant operators, such as (relatvistic and non-relativistic) Schr\"odinger operators. To the authors best knowledge these results have not previously appeared in the literature.

We start by giving a counterexample which shows that a naive generalization of our main result cannot be true. This example was already discussed in \cite{janas1}. 

\begin{ex}\label{exa:counter}
    \cref{thm:main} does not hold if one only asks for the second-order derivatives of $f$ or $\widetilde f^{(s)}$ to have bounded oscillation.
    In \cite[Ex.\ 3.6]{janas1} it is proved that the Toeplitz operator $T_{\Re z^3}$ is not essentially self-adjoint on the domain $\{f\in\H: f(z) \cdot\Re z^3 \in L^2(\CC,\mu) \}$.
    Since this domain contains $\SE$ (cf.\ \cref{thm:schwartz_elts}.\ref{it:bargmann}) it follows that $T_{\Re z^3}$ is not essentially self-adjoint on any of the domains $\SE$, $\mathcal P$ or $\mathcal E$.
\end{ex}

\cref{thm:forms_schrodinger} has immediate applications to Schr\"odinger operators because $\partial_j (-\Delta + V(x)) = \partial_{x_j}V(x)$ and $\partial_{j+n}(-\Delta+ V(x) )= 2i \partial_{x_j}$ for sufficiently regular potentials $V:\RR^n\to\RR$. 
We may even include operator-valued potentials, which describe an interaction between the canonical degrees of freedom (i.e., position and momentum) and inner degrees of freedom such as the particle's spin.

\begin{ex}\label{thm:schroedinger}
    Let $V : \RR^n \to \BO\K_h$ be differentiable with derivatives of bounded oscillation. 
    Then the non-relativistic and relativistic Schr\"odinger operators 
    \[
        -\Delta + V(x) \quad \text{ and }\quad \sqrt{-\Delta +m^2} + V(x)\1_2
    \] 
    with $m>0$ are essentially self-adjoint operators  $\Sch(\RR^n;\K) \to L^2(\RR^n;\K)$. 
\end{ex}

Another type of interaction with inner degrees of freedom is through magnetic fields.
Our theorems do, however, not directly apply to Schr\"odinger operators with magnetic fields or, more generally, to Pauli operators.
For a potential $V(x)$ and a magnetic field $\vec B$ with vector potential $\vec A$ (i.e., $\nabla\times \vec A = \vec B$), the Pauli operator is of the form 
\[
    H_{\text{Pauli}} = \big( it\nabla - \vec A \big)^2 \,\1_2- \vec B \cdot \vec \sigma + V(x):\Sch(\RR^3;\CC^2) \to L^2(\RR^3;\CC^2),
\] 
where $\vec \sigma = (\sigma_x,\sigma_y,\sigma_z)$ is the vector of Pauli matrices and $\vec B \cdot \vec \sigma \psi (x)= \sum B_j(x) \sigma_j \psi(x)$.
We can however use the relaxed assumptions in \cref{thm:magnetic} to show essential self-adjointness of the Toeplitz version of the Pauli operator:

\begin{ex}\label{exa:magnetic}
    Let $\vec A :\RR^3 \to \RR^3$ be a $C^2$-function with bounded derivatives and $V(x)$ a potential as in \cref{thm:schroedinger}.
    Then the Toeplitz operator $T_f : \SE(\CC^2) \to \SE(\CC^2)$ with
    \[
        f(x+i\xi) = |\xi - \vec A(x)|^2 \1_2 - \vec B(x) \cdot \vec\sigma + V(x)\1_2,
    \] 
    where $x,y\in\RR^n$,
    is essentially self-adjoint in $\H(\CC^2)$.
\end{ex}

\begin{proof}
    Expanding yields $f(x+i\xi) = \abs\xi^2\1_2 - 2\xi\cdot \vec A(x) + \abs{\vec A(x)}^2\1_2 - t \vec B(x) \cdot \vec\sigma$.
    We check the assumptions discussed in \cref{thm:magnetic}:
    The symbol is obviously quadratically bounded.
    Furthermore, $\partial_\theta f = x \cdot \nabla_\xi f - \xi \cdot \nabla_x f$ is clearly linearly bounded.
\end{proof}

\begin{ex}
    Put $\K= \ell^2(\mathbb N)$.
    Let $M \in\BO\K$ be the multiplication operator $M(x_n)_n = (-n^{-2} x_n)_n$.
    Denote by $S$ the right shift, i.e., $(Sx)_n = x_{n-1}$ with $x_0\coloneqq  0$.
    Now consider the function $f : \CC \to \BO{\ell^2(\NN)}$ given by
    \[
        f(z) = g(z)\1 + \alpha(z S + \bar z S^*) + M
    \] 
    for some $\alpha>0$ and some $g:\CC\to\RR$ that satisfies the assumption of \cref{thm:main}.
    Then $T_f$ is essentially self-adjoint.
    For $g= \frac12\abs z^2$ this can be used to model the interaction of a quantum harmonic oscillator with the bounded states of an atom.
\end{ex}

\begin{ex}
    Put $\K= L^2(\CC)$ .
    Let $f:\CC^2 \to \RR$ be measurable and assume that $f(\placeholder,w)$ is differentiable with derivatives of bounded oscillation for almost all $w$.
    Furthermore, assume that $f(z,\placeholder) \in L^\infty$ for all $z$.
    Then the Toeplitz operator $T_F$ with symbol $F(z) = M_{f(z,\placeholder)}$, i.e., $[F(z)g](w) = f(z,w) g(w)$, is essentially self-adjoint.
\end{ex}

\section{Outlook}

We want to discuss the \emph{problem of dynamical completeness} for the Toeplitz quantization. 

Let $f : \RR^{2n} \to \RR$ be continuously differentiable on phase space. Then $f$ is said to be  \emph{classically complete} (up to a null-set of initial values) if the Hamiltonian dynamics generated by $f$, i.e., the flow on $\RR^{2n}$ generated by the ODE
\begin{equation}\label{eq:hamiltons_eq}
    (\dot x,\dot \xi) = -J \nabla f(x,\xi),
\end{equation} 
where $J(x,\xi)= (\xi,-x)$, exists for all times and (almost) all initial values.
Existence of the flow for almost all initial values and all times is sufficient to define an isometric one-parameter group on $L^1(\RR^{2n})$ where isometricity follows from Liouvilles theorem in symplectic geometry \cite[Prop.\ 3.3.4]{abraham2008foundations}.
A symmetric operator $H$ on a Hilbert space is said to be \emph{quantum complete} if it is essentially self-adjoint.

Let $\mathcal Q$ be some quantization scheme which maps a classical Hamiltonian function $f$ on phase space to a symmetric operator $\mathcal Q(f)$ on some Hilbert space $\H$. Examples are the Berezin-Toeplitz quantization $f\mapsto T_f$ and the Weyl-quantization which sends a symbol $f$ to its Weyl-pseudodifferential operator $op^w(f)$. 
The problem of dynamical completeness asks whether quantum completeness of $\mathcal Q(f)$ is equivalent to classical completeness of $f$. 

For the Weyl-quantization well-known counterexamples in both directions exist already on the level of Schr\"odinger operators.
See, for instance, the examples on $L^2(\RR_+)$ in \cite[Chap.\ X.1]{reed1975methods} which can be easily extended to the whole real line by considering the potentials $V(\abs x)$. 

What makes these counterexamples work is the fact that the chosen potentials $V(x)$ become ``more singular'' as $x\to\infty$ but it is straightforward to see that the counterexamples fail if one considers heat-transformed potentials $\widetilde V^{(t)}(x)$ instead, even for arbitrarily short times $t>0$.
This, however, is precisely what happens if one uses the Toeplitz quantization $T_f$ with $f(x,\xi) = \xi^2 + V(x)$ instead of the Schr\"odinger operator $-\Delta + V(x)$ because  $T_f \equiv -\Delta + \widetilde V^{(t/2)}(x) +\frac t2 \1$ for measurable and polynomially bounded $V$.
This leads us to the following conjecture:

\begin{con}
    For a reasonably regular Hamiltonian function $f :\RR^{2n}\equiv \CC^n \to\RR$, the following are equivalent
    \begin{enumerate}[(1)]
        \item 
            The Hamiltonian function $f$ is classically complete up to a null-set of inital values.
        \item 
            The Berezin-Toeplitz quantization $T_f$ is quantum complete for all sufficiently small quantization parameters $t>0$.
    \end{enumerate} 
\end{con}

\cref{thm:main} implies that this is true for the class of differentiable symbols with {\nobreak Lipschitz} continuous derivatives but, of course, this class consists only of classically complete functions.

It is not hard to generate counterexamples if one would ask for classical completeness to hold everywhere: For $n>1$ one can pick $f \in L^1$ such that classical completeness fails on a null-set. Still the corresponding quantization $T_f$ is a hermitian trace-class operator and thus self-adjoint.

\subsection*{Acknowledgments}

We thank Robert Fulsche for helpful discussions and suggestions. The second named author acknowledges support by the Quantum Valley Lower Saxony.

\bibliographystyle{abbrvArXiv}
\bibliography{main.bbl}
\end{document}